\documentclass[12pt]{article}
\usepackage{amsmath,amsthm,amssymb,amscd,ascmac,amsfonts,multicol,tabularx}

\newtheorem{main thm}{\quad Theorem}
\newtheorem{df}{\quad Definition}[section]
\newtheorem{thm}[df]{\quad Theorem}
\newtheorem{lem}[df]{\quad Lemma}

\newtheorem{prop}[df]{\quad Proposition}
\newtheorem{rem}[df]{\quad Remark}
\newtheorem{eg}[df]{\quad Example}

\newcommand{\Z}{\mathbb{Z}}
\newcommand{\Q}{\mathbb{Q}}

\newcommand{\F}{\mathbb{F}}

\newcommand{\alp}{\alpha}

\newcommand{\isom}{\simeq}

\newcommand{\x}{\times}

\newcommand{\up}{\stackrel}

\DeclareMathOperator*{\zetaprod}{{\displaystyle{\prod\kern-1.45em\coprod}}}

\newtheorem{nonumberthm}{\quad Theorem 1.2rev}

\begin{document}
\begin{center}
{\Large {Note on families of pairing-friendly elliptic curves with small embedding degree}
}\\
\medskip{Keiji OKANO}
\end{center}


\begin{abstract}
Pairing-based cryptographic schemes require so-called {\it pairing-friendly elliptic curves}, which have special properties.
The set of pairing-friendly elliptic curves that are generated by given polynomials 
form a {\it complete family}.
Although a complete family with a $\rho$-value of $1$ is the ideal case, there is only one such example that is known; this was given by Barreto and Naehrig (Lecture Notes in Computer Science, 3897, Springer, Berlin, 2006, pp. 319--331).
We prove that there are no ideal families with embedding degree $3$, $4$, or $6$ and that many complete families with embedding degree $8$ or $12$ are nonideal, even if we chose noncyclotomic families.
\end{abstract}


\section{Introduction}\label{intro}

Pairing-based cryptographic schemes, which have been suggested independently by Boneh-Franklin \cite{BF} and Sakai-Ohgishi-Kasahara \cite{SOK00}, are based on pairings on elliptic curves.
They fit many novel protocols for which no other practical implementation is known;
see \cite{BF}, \cite{Joux}, and \cite{SOK00} for pioneering works in this field, and see \cite{FST} and \cite{Paterson} for surveys.
One of the features of these schemes is that they require so-called 
{\it pairing-friendly elliptic curves}, 
which have special properties,
whereas the elliptic ElGamal cryptosystems can be implemented by using almost randomly generated elliptic curves.
More precisely, a pairing-friendly elliptic curve $E$ over a finite field $\F_q$ has the following properties (see \cite{FST}):
(i) The curve has a subgroup $G$ of large prime order $r$ such that 
$r \mid q^k-1$ for some integer $k$, and 
$r \nmid q^i-1$ for $0<i<k$.
(ii) The parameters $q$, $r$, and $k$ should be chosen such that $r \ge \sqrt{q}$ and $k<(\log_2 r)/{8}$, 
which mean that the discrete logarithm problem is not only infeasible both in $G$ and $\F_q^\x$ but also the pairings can be computed.
Here, $k$ is called the {\it embedding degree} of $E$ with respect to $r$.
In other words, the embedding degree is the degree of the extension field over $\F_q$ to which the pairing maps.
Also, we define the $\rho$-value of $E$ as
$
\rho(E)=({\log q})/({\log r}).
$
In general, curves with small $\rho$-values are desirable in order to speed up the arithmetic on elliptic curves.
If $\rho(E)=1$, then the curve is ideal.
Although there are some methods for constructing pairing-friendly elliptic curves,
it is known to be very rare that $\rho$ takes the value $1$.
Recently, it has been found that such ideal curves are also required in the theory of zk-SNARK \cite{BCTV} (also see \cite{CFHKKNPZ}), 
which is an application for which these had not been assumed.

In practice, we need to construct curves of a specified bit size for each of various embedding degrees.
To this end, in order to describe the families of pairing-friendly elliptic curves, we give four polynomials $t(x)$, $r(x)$, $q(x)$, and $y(x)$, 
instead of the trace $t$ of the Frobenius map on $E$, the above two parameters $r$ and $q$.
The families of pairing-friendly elliptic curves 
that are generated by such polynomials 
are called {\it complete families}, and we consider them
in this paper.
Moreover, when $r(x)$ is chosen to be a cyclotomic polynomial, this yields the most popular complete family, known as the cyclotomic case
(see Section \ref{ps-friendly} for more details).

We also define the {\it $\rho$-value} $\rho(t, r, q)$ for complete families (Definition \ref{family of p-f cv}).
The case where the $\rho$-value equals $1$ is also ideal, while this is not true for most cases.
We thus consider the question, under what 
conditions is the $\rho$-value equal to or close to $1$?
There is only one known example of a complete family with $\rho(t, r, q)=1$; it was constructed by Barreto and Naehrig \cite{BN} (Remark \ref{BN}). 
Various approaches have been used to search for the ideal case; see, for example, 
\cite{FST}, \cite{LL}, \cite{KSS}, \cite{Okano2012}, \cite{Min Sha}, and \cite{Yoon}.

In this paper, 
we show that there is {\it no} ideal case for which $k=3$, $4$, or $6$.
In addition, we consider the cases of $k=8$ and $12$, including noncyclotomic cases.

\begin{thm}\label{main thm 1}
\makeatletter
  \parsep   = 0pt
  \labelsep = 5pt
  \def\@listi{%
     \leftmargin = 20pt \rightmargin = 0pt
     \labelwidth\leftmargin \advance\labelwidth-\labelsep
     \topsep     = 0\baselineskip
     \partopsep  = 0pt \itemsep       = 0pt
     \itemindent = 0pt \listparindent = 10pt}
  \let\@listI\@listi
  \@listi
  \def\@listii{%
     \leftmargin = 20pt \rightmargin = 0pt
     \labelwidth\leftmargin \advance\labelwidth-\labelsep
     \topsep     = 0pt \partopsep     = 0pt \itemsep   = 0pt
     \itemindent = 0pt \listparindent = 10pt}
  \let\@listiii\@listii
  \let\@listiv\@listii
  \let\@listv\@listii
  \let\@listvi\@listii
  \makeatother
Let $k=3$, $4$, or $6$, and let $D$ be a square-free positive integer.
Suppose that $(t(x), r(x), q(x))$ parameterizes a complete family of elliptic curves with complex multiplication (CM) discriminant $D$ and embedding degree $k$.
Then, $\rho(t, r, q) \neq 1$.
\end{thm}

\begin{thm}\label{main thm 3}
\makeatletter
  \parsep   = 0pt
  \labelsep = 5pt
  \def\@listi{%
     \leftmargin = 20pt \rightmargin = 0pt
     \labelwidth\leftmargin \advance\labelwidth-\labelsep
     \topsep     = 0\baselineskip
     \partopsep  = 0pt \itemsep       = 0pt
     \itemindent = 0pt \listparindent = 10pt}
  \let\@listI\@listi
  \@listi
  \def\@listii{%
     \leftmargin = 20pt \rightmargin = 0pt
     \labelwidth\leftmargin \advance\labelwidth-\labelsep
     \topsep     = 0pt \partopsep     = 0pt \itemsep   = 0pt
     \itemindent = 0pt \listparindent = 10pt}
  \let\@listiii\@listii
  \let\@listiv\@listii
  \let\@listv\@listii
  \let\@listvi\@listii
  \makeatother
Let $k=8$ or $12$, and let $D$ be a square-free positive integer.
Suppose that $(t(x), r(x), q(x))$ parameterizes a complete family of elliptic curves with CM discriminant $D$ and embedding degree $k$.
If 
$$
\sqrt{-D} \in \Q(\zeta_k)
\text{\ \ and\ \ }
\deg r(x) \neq 2 \deg t(x),
$$ 
then $\rho(t, r, q) \neq 1$.
Here, $\zeta_k$ is a primitive $k$th root of unity.
\end{thm}


In the following, we denote the $k$th cyclotomic polynomial and Euler's totient function as $\Phi_k(x)$ and $\varphi(x)$, respectively.
Note that it is known that $r(x)$ is a factor of $\Phi_k(t(x)-1)$, and 
the degree of $r(x)$ is a multiple of $\varphi(k)$.
Moreover, note that if $\deg t(x)=2$, then $\Phi_8(t(x)-1)$ is irreducible in the case where $k=8$ (see \cite{GMV}), and, in the case $k=12$, a family with $\rho(t, r, q)=1$ is given by \cite{BN}.



\section{Family of Pairing-Friendly Elliptic Curves}\label{ps-friendly}

In this section, we briefly explain the strategy for constructing complete families of pairing-friendly elliptic curves as proposed by Brezing and Weng \cite{BW}.
In the following, we use the notation $\Z$ and $\Q$ for the set of rational integers and rational numbers, respectively.
For an elliptic curve $E/\F_q$, define $t$ by the trace of the Frobenius map on $E$.
Then, the order of $E(\F_q)$ is described as 
$\# E(\F_q)=q+1-t$.

\begin{df}\label{def of embed deg}
Let $E$ be an elliptic curve over $\F_q$.
Suppose that $E(\F_q)$ has a subgroup of order $r$ with 
$\gcd(r,q)=1$.
The embedding degree of $E$ with respect to $r$ is the extension degree
$[\F_q(\mu_r):\F_q]$.
Here, $\mu_r$ is the group of all $r$th roots of unity in an algebraic closure of $\F_q$.
\end{df}

\begin{prop}\label{equiv of embed deg}
{\rm (\cite[Remark 2.2 and Proposition 2.4]{FST})}
\makeatletter
  \parsep   = 0pt
  \labelsep = 5pt
  \def\@listi{%
     \leftmargin = 20pt \rightmargin = 0pt
     \labelwidth\leftmargin \advance\labelwidth-\labelsep
     \topsep     = 0\baselineskip
     \partopsep  = 0pt \itemsep       = 0pt
     \itemindent = 0pt \listparindent = 10pt}
  \let\@listI\@listi
  \@listi
  \def\@listii{%
     \leftmargin = 20pt \rightmargin = 0pt
     \labelwidth\leftmargin \advance\labelwidth-\labelsep
     \topsep     = 0pt \partopsep     = 0pt \itemsep   = 0pt
     \itemindent = 0pt \listparindent = 10pt}
  \let\@listiii\@listii
  \let\@listiv\@listii
  \let\@listv\@listii
  \let\@listvi\@listii
  \makeatother
Assume that $r \mid \# E(\F_q)$ is prime relative to $q$. 
Then, the following two conditions are equivalent:
\begin{enumerate}
\item[{\rm (i)}]
$E$ has embedding degree $k$ with respect to $r$.
\item[{\rm (ii)}]
$k$ is the smallest positive integer such that 
$r \mid q^k-1$.
\end{enumerate}
Moreover, if $r$ is prime such that $r \nmid kq$, then $E$ has embedding degree $k$ with respect to $r$ if and only if 
$r \mid \Phi_k(t-1)$.
\end{prop}

We now describe the CM method as proposed by Atkin and Morain \cite{AM}, which is a strategy for constructing elliptic curves with given parameters.

\begin{thm}\label{AM}{\rm (Atkin and Morain \cite{AM})}
\makeatletter
  \parsep   = 0pt
  \labelsep = 5pt
  \def\@listi{%
     \leftmargin = 20pt \rightmargin = 0pt
     \labelwidth\leftmargin \advance\labelwidth-\labelsep
     \topsep     = 0\baselineskip
     \partopsep  = 0pt \itemsep       = 0pt
     \itemindent = 0pt \listparindent = 10pt}
  \let\@listI\@listi
  \@listi
  \def\@listii{%
     \leftmargin = 20pt \rightmargin = 0pt
     \labelwidth\leftmargin \advance\labelwidth-\labelsep
     \topsep     = 0pt \partopsep     = 0pt \itemsep   = 0pt
     \itemindent = 0pt \listparindent = 10pt}
  \let\@listiii\@listii
  \let\@listiv\@listii
  \let\@listv\@listii
  \let\@listvi\@listii
  \makeatother
Let $k$ be a positive integer.
Suppose that there are some $t,r,q$ satisfying the following properties:
\begin{enumerate}
\item[{\rm (i)}]
$r$ is a prime number.
\item[{\rm (ii)}]
$q$ is a power of a prime number.
\item[{\rm (iii)}]
$r \mid q+1-t$ and $\gcd(t,q)=1$.
\item[{\rm (iv)}]
$r \mid q^k-1$ and $r \nmid q^i-1$ for all $1 \le i <k$.
\item[{\rm (v)}]
There exist some $y \in \Z$ and some square-free positive integer $D$ such that an equation
$Dy^2=4q-t^2$
holds
{\rm (}$D$ is called a CM discriminant{\rm )}.
\end{enumerate}
Then, there exists an ordinary elliptic curve $E$ over $\F_q$ that satisfies the following:
\begin{enumerate}
\item[{\rm (a)}]
The order $\# E(\F_q)$ of $E(\F_q)$ is 
$\# E(\F_q) = q+1-t$, 
and there is a subgroup of $E(\F_q)$ with prime order $r$.
\item[{\rm (b)}]
The embedding degree with respect to $r$ is $k$.
\end{enumerate}
\end{thm}

For applications, it is necessary to be able to construct curves of a specified bit size.
To this end, to describe the families of pairing-friendly curves, we give the four parameters 
$t$, $r$, $q$, and $y$ in Theorem \ref{AM}, and these are given as the polynomials $t(x)$, $r(x)$, $q(x)$, and $y(x)$ with a parameter $x$.
According to \cite{FST}, we introduce the following definition, which is based on the conjecture of Bouniakowski and Schinzel (see \cite[p. 323]{Lang}).

\begin{df}\label{rep primes}
\makeatletter
  \parsep   = 0pt
  \labelsep = 5pt
  \def\@listi{%
     \leftmargin = 20pt \rightmargin = 0pt
     \labelwidth\leftmargin \advance\labelwidth-\labelsep
     \topsep     = 0\baselineskip
     \partopsep  = 0pt \itemsep       = 0pt
     \itemindent = 0pt \listparindent = 10pt}
  \let\@listI\@listi
  \@listi
  \def\@listii{%
     \leftmargin = 20pt \rightmargin = 0pt
     \labelwidth\leftmargin \advance\labelwidth-\labelsep
     \topsep     = 0pt \partopsep     = 0pt \itemsep   = 0pt
     \itemindent = 0pt \listparindent = 10pt}
  \let\@listiii\@listii
  \let\@listiv\@listii
  \let\@listv\@listii
  \let\@listvi\@listii
  \makeatother
Let $f(x)$ be a polynomial in $\Q[x]$.
\begin{enumerate}
\item[{\rm (i)}]
If there is some $a \in \Z$ such that $f(a) \in \Z$, then we say that $f(x)$ represents integers.
\item[{\rm (ii)}]
Assume that $f(x)$ is nonconstant, irreducible, and represents integers.
If $f(x)$ has a positive leading coefficient and 
\[
{\rm gcd}(\{ f(x) |\; x\ such\ that\ f(x) \in \Z\})=1,
\]
then we say that $f(x)$ represents primes.
\end{enumerate}
\end{df}

Bouniakowski, Schinzel, and others conjectured that if $f(x)$ represents primes, then $f(x)$ has infinitely many prime values.

\begin{df}\label{family of p-f cv}
\makeatletter
  \parsep   = 0pt
  \labelsep = 5pt
  \def\@listi{%
     \leftmargin = 20pt \rightmargin = 0pt
     \labelwidth\leftmargin \advance\labelwidth-\labelsep
     \topsep     = 0\baselineskip
     \partopsep  = 0pt \itemsep       = 0pt
     \itemindent = 0pt \listparindent = 10pt}
  \let\@listI\@listi
  \@listi
  \def\@listii{%
     \leftmargin = 20pt \rightmargin = 0pt
     \labelwidth\leftmargin \advance\labelwidth-\labelsep
     \topsep     = 0pt \partopsep     = 0pt \itemsep   = 0pt
     \itemindent = 0pt \listparindent = 10pt}
  \let\@listiii\@listii
  \let\@listiv\@listii
  \let\@listv\@listii
  \let\@listvi\@listii
  \makeatother
Let $k$ be a positive integer, and let $D$ be a positive square-free integer.
Suppose that a triple of nonzero polynomials $(t(x),r(x),q(x)) \in \Q[x]^{3}$ satisfies the following conditions:
\begin{enumerate}
\item[{\rm (i)}] 
$r(x)$ represents primes.
\item[{\rm (ii)}] 
$q(x)$ represents primes.
\item[{\rm (iii)}] 
$r(x) \mid q(x)+1-t(x)$, i.e., there exists $h(x) \in \Q[x]$ such that $h(x)r(x)=q(x)+1-t(x)$.
\item[{\rm (iv)}]
$r(x) \mid \Phi_k(t(x)-1)$.
\item[{\rm (v)}] 
There is some $y(x) \in \Q[x]$ such that
$$
Dy(x)^2
=
4q(x)-t(x)^2
=
4h(x)r(x)-(t(x)-2)^2.
$$
\end{enumerate}
Then, we say that $(t(x),r(x),q(x))$ parameterizes a complete family of pairing-friendly elliptic curves with embedding degree $k$ and CM discriminant $D$.
Moreover, we define
$$
\rho(t,r,q):=
\frac{\deg q(x)}{\deg r(x)}=
\frac{2 \max\{ \deg y(x), \deg t(x)\}}{\deg r(x)}.
$$
\end{df}

If $(t(x),r(x),q(x))$ parameterizes a complete family of pairing-friendly elliptic curves with embedding degree $k$,
then $r(x)$ defines a field $\Q[x]/(r(x))$ that is isomorphic to a field containing the $k$th cyclotomic field and the imaginary quadratic field $\Q(\sqrt{-D})$, 
and by Definition \ref{family of p-f cv} (iv)(v), $t(x)-1$ corresponds to a primitive $k$th root of unity.

Next, we describe the Brezing-Weng method, which is a generalization of the Cocks-Pinch method \cite{CP}.

\begin{thm}[Brezing-Weng \cite{BW}]\label{BW}
\makeatletter
  \parsep   = 0pt
  \labelsep = 5pt
  \def\@listi{%
     \leftmargin = 20pt \rightmargin = 0pt
     \labelwidth\leftmargin \advance\labelwidth-\labelsep
     \topsep     = 0\baselineskip
     \partopsep  = 0pt \itemsep       = 0pt
     \itemindent = 0pt \listparindent = 10pt}
  \let\@listI\@listi
  \@listi
  \def\@listii{%
     \leftmargin = 20pt \rightmargin = 0pt
     \labelwidth\leftmargin \advance\labelwidth-\labelsep
     \topsep     = 0pt \partopsep     = 0pt \itemsep   = 0pt
     \itemindent = 0pt \listparindent = 10pt}
  \let\@listiii\@listii
  \let\@listiv\@listii
  \let\@listv\@listii
  \let\@listvi\@listii
  \makeatother
Let $k$ be a positive integer, and let $D$ be a positive square-free integer.
Then, execute the following steps.
\begin{enumerate}
\item
Choose an algebraic number field $K$ that contains the $k$th cyclotomic field and $\Q(\sqrt{-D})$.
\item
Find an irreducible polynomial $r(x) \in \Z[x]$ with positive leading coefficient and an isomorphism such that $\Q[x]/(r(x)) \up{\sim}{\to} K$.
\item
Let $t(x)-1 \in \Q[x]$ be a polynomial mapping to a fixed $k$th root of unity $\zeta_k \in K$ by the above isomorphism.
\item
Let $y(x) \in \Q[x]$ be a polynomial mapping to $\frac{\zeta_k-1}{\sqrt{-D}} \in K$ by the above isomorphism.
\item
Let $q(x) \in \Q[x]$ be given by
$q(x):=\frac{1}{4}(t(x)^2+Dy(x)^2)$.
\end{enumerate}
If $q(x)$ and $r(x)$ represent primes,
then the triple $(t(x),r(x),q(x))$ parameterizes a complete family of elliptic curves with embedding degree $k$ and CM discriminant $D$.
We note that $t(x)$, $y(x)$ are determined up to modulus $r(x)$.
\end{thm}

\begin{rem}\label{BN}
{\rm The choice of $r(x)$ is an important part of this algorithm.
When $r(x)$ is chosen to be a cyclotomic polynomial, this yields the most popular complete family;
this is called the cyclotomic case.
Several examples are collected in \cite{FST} and \cite{KSS}.
Barreto and Naehrig \cite{BN} gave an example of $\rho(t,r,q)=1$ with $k=12$ and $D=3$:
\begin{eqnarray*}
&
\begin{cases}
t(x)=6x^2+1,\ 
\\
r(x)=36x^4+36x^3+18x^2+6x+1,\ 
\ \ 
\\
q(x)=36x^4+36x^3+24x^2+6x+1.
\end{cases}
&
\end{eqnarray*}
This is the {\it only} known example of $(t(x),r(x),q(x))$ that parameterizes a complete family of curves with a $\rho$-value of $1$.
}
\end{rem}



\section{Proof of Theorem \ref{main thm 1}}
First, we show Theorem \ref{main thm 1} for the case where $\sqrt{-D} \in \Q(\zeta_k)$.
The proof is similar to the proof shown in \cite[Proposition 4.1]{Okano2012}.
Let $\zeta_k$ be a primitive $k$th root of unity corresponding to $t(x)-1$ under a fixed isomorphism 
$\Q[x]/(r(x)) \isom K \supset \Q(\zeta_k)$.
For simplicity, put $X=t(x)-1$.
Assume that 
$\deg t(x)<\deg r(x)$ and $\deg y(x)<\deg r(x)$.

Suppose that $k=4$.
Then, $D=1$, and so $\sqrt{-1}$ corresponds to $s(x)=\pm X$.
Therefore, we have
\begin{eqnarray*}
y(x) 
&\equiv&
\frac{(X-1)s(x)}{-1}
	=\mp (X^2-X)
\\
&\equiv&
\pm (X+1)
\mod r(x),
\end{eqnarray*}
since $r(x) \mid \Phi_4(X)=X^2+1$.
Hence,
$q(x) 
=
\frac{1}{4}\left(t(x)^2+Dy(x)^2\right)
=\frac{1}{2}(X+1)^2.
$
This contradicts the assumption that $q(x)$ represents primes.
Therefore, 
$\deg t(x) \ge \deg r(x)$ or $\deg y(x) \ge \deg r(x)$, 
and so $\rho(t,r,q) \ge 2$ if $k=4$.
In the same way, we obtain $\rho(t,r,q) \ge 2$ for the cases $k=3$ and $6$:

Suppose that $k=3$.
Then $D=3$, and so $\sqrt{-3}=\pm (2\zeta_3+1)$ corresponds to $s(x)=\pm (2X+1)$.
Therefore,
\begin{eqnarray*}
y(x) 
&\equiv&
\frac{(X-1)s(x)}{-3}
	=\mp \frac{2X^2-X-1}{3}
\\
&\equiv&
\pm (X+1)
\mod r(x),
\end{eqnarray*}
since $r(x) \mid \Phi_3(X)=X^2+X+1$.
Hence,
$
q(x) 
=
(X+1)^2.
$
This contradicts the assumption that $q(x)$ represents primes.

Suppose that $k=6$.
Then $D=3$, and so $\sqrt{-3}=\pm (2\zeta_6-1)$ corresponds to $s(x)=\pm (2X-1)$.
Therefore,
\begin{eqnarray*}
y(x) 
&\equiv&
\frac{(X-1)s(x)}{-3}
	=\mp \frac{2X^2-3X+1}{3}
\\
&\equiv&
\pm \frac{1}{3}(X+1)
\mod r(x),
\end{eqnarray*}
since $r(x) \mid \Phi_6(X)=X^2-X+1$.
Therefore,
$
q(x) 
=
\frac{1}{3}(X+1)^2.
$
This contradicts the assumption that $q(x)$ represents primes.

\begin{rem}
These conditions of $r(x)$, $t(x)$, and $q(x)$ in the proof mean that the family provides supersingular elliptic curves.
\end{rem}



Next, we show the case where $\sqrt{-D} \notin \Q(\zeta_k)$.
Assume that $\rho(t,r,q)=1$.
Put $X=t(x)-1$ and $m=\deg t(x)$.
Note that $m>1$, since we assume that $\sqrt{-D} \notin \Q(\zeta_k)$.
Then, we may assume that
$$
r(x)=\Phi_k(X),
$$ 
since $r(x) \mid \Phi_k(X)$ and $1=\rho(t,r,q) \ge \frac{2\deg t(x)}{\deg r(x)}$.

\begin{lem}\label{deg y}
If $k=3$, $4$, or $6$ and $\rho(t,r,q)=1$, then
$
\deg y(x)=\frac{m}{2}.
$
\end{lem}

\begin{proof}
Since $\deg X=m$, the $\Q$-vector space $\Q[x]_{2m-1}$, which consists of all polynomials with degree less than $2m=\deg r(x)$, has bases consisting of
$$
\begin{matrix}
x^{m-1}X, & x^{m-2}X, & \ldots, & xX, & X,
\\
x^{m-1}, & x^{m-2}, & \ldots, & x, & 1.
\end{matrix}
$$
We consider the polynomial $(X-1)s(x)$ that is congruent with $-Dy(x)$ modulo $r(x)$.
We can write $s(x)\in \Q[x]_{2m-1}$ uniquely as
$$
s(x)=
(F_1(x)x+a_1)X+(F_2(x)x+a_2)
$$
for some $a_i \in \Q$ and $F_i(x)$ that satisfy $\deg F_i(x) \le m-2$ ($i=1,2$).
Therefore,
\begin{eqnarray}\label{(X-1)s(x)_1}
&&
(X-1)s(x)
\nonumber
\\
&&
=
F_1(x)xX^2+a_1X^2+(F_2(x)-F_1(x))xX+(a_2-a_1)X-(F_2(x)x+a_2).
\end{eqnarray}
Note that $F_1(x)xX^2$, $a_1X^2$, and $(F_1(x)-F_2(x))xX$ do not have any terms that are of the same degree as those of the others.
On the other hand, dividing $(X-1)s(x)$ by $r(x)=\Phi_k(X)$, there are some $G(x) \in \Q[x]$ and $b \in \Q$ such that
\begin{eqnarray}\label{(X-1)s(x)_2}
(X-1)s(x)=
-Dy(x)+(G(x)x+b)\Phi_k(X).
\end{eqnarray}
Here, $\deg G(x) \le m-2$, since $\deg \left( (X-1)s(x) \right) <m+2m=3m$ and $\deg y(x)<2m$.

If $k=4$, then $\Phi_4(X)=X^2+1$.
Thus, the right-hand side of (\ref{(X-1)s(x)_2}) becomes
\begin{eqnarray}\label{(X-1)s(x)_2.4}
-Dy(x)+G(x)xX^2+bX^2+G(x)x+b.
\end{eqnarray}
Since $\deg y(x)\le m$, the terms in (\ref{(X-1)s(x)_1}) and (\ref{(X-1)s(x)_2.4}) of degree greater than $m$ coincide.
Therefore, we obtain
$G(x)=F_1(x)=F_2(x)$, $a_1=b$.
Thus,
$$
-Dy(x)=
(a_2-a_1)X-2F_1(x)x-(a_2+a_1),
$$
and so 
\begin{eqnarray}\label{D^2y(x)^2 k=4}
&&
D^2y(x)^2
\\
&&
=
-4(a_2-a_1)F_1(x)xX+4F_1(x)^2x^2+4(a_2+a_1)F_1(x)x+(\text{polynomial of $X$}).
\nonumber 
\end{eqnarray}
Since $r(x)=\Phi_4(X) \in \Q[X]$, we see that
$Dy(x)^2=4hr(x)-(X-1)^2 \in \Q[X]$.
If we combine this with (\ref{D^2y(x)^2 k=4}) and $\deg X \ge 2$, we obtain
$$
F_1(x)=0
\text{\ \ or\ \  }
a_1=a_2.
$$
Otherwise, the leading term of $-4(a_2-a_1)F_1(x)xX$ does not vanish, which contradicts $D^2y(x)^2 \in \Q[X]$.
Assume that $F_1(x)=0$.
Then, $F_2(x)=0$, and so $s(x)=a_1X+a_2 \in \Q[X]$.
This implies that $\sqrt{-D} \in \Q(\zeta_k)$, which contradicts our assumption.
Hence, we have $F_1(x) \neq 0$ and $a_1=a_2=b$;
we also have 
$
-Dy(x)=-2F_1(x)x-2a_1.
$ 
Thus,
$$
\deg \left( y(x)^2 \right) \le 2(m-1).
$$
By $y(x)^2 \in \Q[X]$ and $\deg X=m$, we have $\deg y(x)=m/2$.
In the same way, we can obtain the same result if $k=3$ or $6$, as follows.

Suppose that $k=3$.
Then $r(x)=\Phi_3(X)=X^2+X+1$, and so the right-hand side of (\ref{(X-1)s(x)_2}) becomes
\begin{eqnarray*}\label{(X-1)s(x)_2.3}
-Dy(x)+G(x)xX^2+bX^2+G(x)xX+bX+G(x)x+b.
\end{eqnarray*}
Comparing this with (\ref{(X-1)s(x)_1}), we obtain
$G(x)=F_1(x)=F_2(x)-F_1(x)$ and $b=a_1$,
since $\deg y(x)\le m$.
Thus,
$$
-Dy(x)=
(a_2-2a_1)X-3F_1(x)x-(a_2+a_1),
$$
and so 
\begin{eqnarray*}
&&
D^2y(x)^2
\\
&&
=
-6(a_2-2a_1)F_1(x)xX+9F_1(x)^2x^2+6(a_2+a_1)F_1(x)x+(\text{polynomial of $X$}).
\end{eqnarray*}
By $D^2y(x)^2 \in \Q(X)$ and $\deg X \ge 2$, we obtain 
$$
F_1(x)=0
\text{\ \ or\ \  }
2a_1=a_2.
$$
If we assume that $F_1(x)=0$, then this contradicts $\sqrt{-D} \notin \Q(\zeta_4)$ in the same way as in the case $k=4$.
Hence, we have $F_1(x) \neq 0$ and $a_2=2a_1$.
We have 
$
-Dy(x)=-3F_1(x)x-3a_1,
$ 
and so we have $\deg y(x)=m/2$ in the same way as in the case $k=4$.

Suppose that $k=6$.
Then, $r(x)=\Phi_6(X)=X^2-X+1$, and the right-hand side of (\ref{(X-1)s(x)_2}) becomes
\begin{eqnarray*}\label{(X-1)s(x)_2.3}
-Dy(x)+G(x)xX^2+bX^2-G(x)xX-bX+G(x)x+b.
\end{eqnarray*}
Comparing this with (\ref{(X-1)s(x)_1}), we obtain
$G(x)=F_1(x)=-F_2(x)+F_1(x)$, $b=a_1$,
since $\deg y(x)\le m$.
Thus,
$$
-Dy(x)=
a_2X-F_1(x)x-(a_2+a_1),
$$
and so 
\begin{eqnarray*}
D^2y(x)^2
=
-2a_2F_1(x)xX+F_1(x)^2x^2+2(a_2+a_1)F_1(x)x+(\text{polynomial of $X$}).
\end{eqnarray*}
By $D^2y(x)^2 \in \Q(X)$ and $\deg X \ge 2$, we obtain 
$$
F_1(x)=0
\text{\ \ or\ \  }
a_2=0.
$$
If we assume that $F_1(x)=0$, then this induces a contradiction in the same way as in the case $k=4$.
Hence $F_1(x) \neq 0$ and $a_2=0$.
We have 
$
-Dy(x)=-F_1(x)x-a_1,
$ 
and so we have $\deg y(x)=m/2$ in the same way as in the case $k=4$.
\end{proof}

Assume that $\rho(t,r,q)=1$.
Then, $\deg y(x)=m/2$ by Lemma \ref{deg y}.
For each $k$, $Dy(x)^2=4h\Phi_k(X)-(X-1)^2$ has degree $m$.
Hence, if $k=3$, $4$, or $6$, then $h=1/4$ and $Dy(x)^2=3X$, $2X$, or $X$, respectively.
Then, $q(x)=\frac{1}{4}(t(x)^2+Dy(x)^2)$ becomes
$$
\frac{1}{4}
(X^2+5X+1),
\ \ 
\frac{1}{4}
(X^2+4X+1),
\ \ 
\text{or\ \ } \frac{1}{4}
(X^2+3X+1)
$$
if $k=3$, $4$, or $6$, respectively.
None of these can represent integers, and so this would contradict the assumption.
Therefore, we have shown that $\rho(t,r,q)$ cannot be $1$.
\hfill $\square$

\begin{eg}\label{k=4,D=2 and k=6,D=1}
{\rm Let $k=4$, $t(x)=x^2+1$, and $r(x)=x^4+1$.
Then, corresponding $x$ to $\zeta_8$, we obtain an isomorphism $\Q[x]/(r(x)) \isom \Q(\zeta_8)$ .
If we choose $D=2$, then $\sqrt{-2} \notin \Q(\zeta_4)$.
Then,
\begin{eqnarray*}
\begin{cases}
s(x)=x+x^3
\\
y(x)=x\ \left(\equiv \frac{(t(x)-2)s(x)}{-2} \mod r(x) \right)
\\
q(x)=\frac{1}{4}(x^4+4x^2+1).
\end{cases}
\end{eqnarray*}

Also, let $k=6$, $t(x)=x^2+1$, $D=1$, and $r(x)=x^4-x^2+1$.
Then, $\sqrt{-1} \notin \Q(\zeta_6)$, and
corresponding $x$ to $\zeta_{12}$, we obtain 
\begin{eqnarray*}
\begin{cases}
s(x)=x^3
\\
y(x)=x\ \left(\equiv \frac{(t(x)-2)s(x)}{-1} \mod r(x) \right)
\\
q(x)=\frac{1}{4}(x^4+3x^2+1).
\end{cases}
\end{eqnarray*}
}
\end{eg}


\section{Proof of Theorem \ref{main thm 3}}
We now show Theorem \ref{main thm 3}.
Let $m=\deg t(x)$.
Note that $r(x)$ satisfies $r(x) \mid \Phi_k(t(x)-1)$, and it is known that $\varphi(k) \mid \deg r(x)$ (see {\cite[Lemma 5.1]{Free}} or {\cite[Theorem 5.1]{FST}}).
Therefore, there is some integer $n \le m$ such that
$$
\deg r(x)=4n.
$$
Furthermore, if $4n<2m$, then we obtain
$
\rho(t,r,q) >1
$
from 
$
\rho(t,r,q) \ge \frac{2\deg t(x)}{\deg r(x)}.
$
Therefore, we may assume that $m \le 2n$.
If we combine the assumption that $\deg r(x) \neq 2 \deg t(x)$ with these facts, then we obtain
\begin{eqnarray}\label{n<m<2m}
n \le m < 2n.
\end{eqnarray}
We may also assume that $\deg y(x) <\deg r(x)$.
Let $\zeta=\zeta_k$ be a primitive $k$th root of unity corresponding to $t(x)-1$ under a fixed isomorphism $\Q[x]/(r(x)) \isom K \supset \Q(\zeta_k)$.
For $\alp \in \Q(\zeta)$, we define $P(\alp)=P(\alp;x) \in \Q[x]$ by the polynomial corresponding to $\alp$ of degree less than $4n$.
For example, 
$P(\zeta)=t(x)-1$, 
$P((\zeta-1)\sqrt{-D})=-Dy(x)$.
Furthermore, we see that 
$P(\alp+\beta)=P(\alp)+P(\beta)$ 
holds for any $\alp,\beta \in \Q(\zeta)$.
Thus, we will assume that $\rho(t,r,q)=1$.
Then, since $2\deg t(x)<\deg r(x)$ by (\ref{n<m<2m}), we obtain
\begin{eqnarray}\label{deg y(x)=2n}
\deg y(x)=2n.
\end{eqnarray}
Moreover, since $\deg \left( P(\zeta)^2 \right)=2m<4n$, again by (\ref{n<m<2m}), we have
$
P(\zeta^2)=P(\zeta)^2.
$
To prove the theorem, we consider $(\zeta-1)\sqrt{-D}$ and $\zeta (\zeta-1)\sqrt{-D}$.

Suppose that $k=8$ and $D=1$.
Then, $\sqrt{-1}=\pm \zeta^2$, and
\begin{eqnarray*}
-Dy(x)
&=&
P\left( \pm(\zeta-1)\zeta^2\right)
=
P\left( \pm(\zeta^3-\zeta^2)\right)
\\
&=&
\pm P(\zeta^3) \mp P(\zeta)^2.
\end{eqnarray*}
From this, (\ref{n<m<2m}), and (\ref{deg y(x)=2n}), we obtain
\begin{eqnarray}\label{P(zeta^3)}
\deg P(\zeta^3) \le 2m.
\end{eqnarray}
On the other hand, by (\ref{n<m<2m}),
we know that $\deg \left( P(\zeta)P((\zeta-1)\sqrt{-1}) \right)=m+2n<4n$.
Since $\zeta^4=-1$, we have 
$$
P\left( \mp(\zeta^3+1) \right)
=
P \left( \zeta (\zeta-1)\sqrt{-1} \right)
=
P(\zeta)P\left((\zeta-1)\sqrt{-1}\right).
$$
In particular, $P(\zeta^3+1)$ has degree $m+2n$.
However, by (\ref{P(zeta^3)}), $\deg P(\zeta^3+1) \le 2m$,
and so $2n \le m$.
This contradicts (\ref{n<m<2m}).
Thus, $\rho(t,r,q) \neq 1$ in this case.

The other cases can be proven in the same way.
Suppose that $k=8$ and $D=2$.
Then, $\sqrt{-2}=\pm (\zeta+\zeta^3)$, and so
\begin{eqnarray*}
-Dy(x)
&=&
P\left( \pm(\zeta-1)(\zeta+\zeta^3)\right)
=
P\left( \pm(-\zeta^3+\zeta^2-\zeta-1)\right)
\\
&=&
\mp P(\zeta^3) \pm P(\zeta^2-\zeta-1).
\end{eqnarray*}
From this, (\ref{n<m<2m}), and (\ref{deg y(x)=2n}), we obtain
$
\deg P(\zeta^3) \le 2m.
$
On the other hand, from the above approach, we know that 
$$
P\left( \zeta (\zeta-1)\sqrt{-2}\right)
=
\pm P(\zeta^3-\zeta^2-\zeta+1)
$$ 
has degree $m+2n$.
However, $\deg P(\zeta^3-\zeta^2-\zeta+1) \le 2m$, and so $2n \le m$, which contradicts (\ref{n<m<2m}).
Thus, $\rho(t,r,q) \neq 1$.

Suppose that $k=12$ and $D=1$.
Then, $\sqrt{-1}=\pm \zeta^3$, and so
\begin{eqnarray*}
-Dy(x)
&=&
P\left( \pm(\zeta-1)\zeta^3\right)
=
P\left( \pm(-\zeta^3+\zeta^2-1)\right)
\\
&=&
\mp P(\zeta^3) \pm P(\zeta^2-1).
\end{eqnarray*}
From this, (\ref{n<m<2m}), and (\ref{deg y(x)=2n}), we obtain
$
\deg P(\zeta^3) \le 2m.
$
On the other hand, since $\zeta^4-\zeta^2+1=0$, from the above approach, we know that
$$
P\left( \zeta (\zeta-1)\sqrt{-1}\right)
=
P\left( \pm(\zeta^3-\zeta^2-\zeta+1)\right)
$$ 
has degree $m+2n$.
However, $\deg P(\zeta^3-\zeta^2-\zeta+1) \le 2m$, and so $2n \le m$.
This contradicts (\ref{n<m<2m}).
Thus, $\rho(t,r,q) \neq 1$.

Finally, suppose that $k=12$ and $D=3$.
Then, $\sqrt{-3}=\pm (2\zeta^2-1)$, and so 
\begin{eqnarray*}
-Dy(x)
&=&
P\left( \pm(\zeta-1)(2\zeta^2-1)\right)
=
P(\pm(2\zeta^3-2\zeta^2-\zeta+1))
\\
&=&
\pm 2P(\zeta^3) \mp P(2\zeta^2+\zeta-1).
\end{eqnarray*}
From this, (\ref{n<m<2m}), and (\ref{deg y(x)=2n}), we obtain
$
\deg P(\zeta^3) \le 2m.
$
On the other hand, from the above approach, we know that 
$$
P(\zeta (\zeta-1)\sqrt{-3})=P(\pm(-2\zeta^3+\zeta^2+\zeta-2))
$$ 
has degree $m+2n$.
However, $\deg P(-2\zeta^3+\zeta^2+\zeta-2) \le 2m$, and so $2n \le m$.
This contradicts (\ref{n<m<2m}).
Thus, we have shown that $\rho(t,r,q) \neq 1$ for all cases.
\hfill $\square$


\section{Conclusion}
We showed that 
there is no ideal case when $k=3$, $4$, or $6$.
We also showed that, if $\sqrt{-D} \in \Q(\zeta_k)$, there is no ideal case when $k=8$ or $12$ except when $\deg r(x) = 2 \deg t(x)$.
Note that the Barreto and Naehrig family satisfies $k=12$ and $\deg r(x) = 2 \deg t(x)$.
The case where $\sqrt{-D} \not\in \Q(\zeta_k)$ has not been sufficiently examined.

\vspace{10pt}

{\bf Acknowledgment}
\\
This work was supported by JSPS KAKENHI Grant Number 26870486, Grant-in-Aid for Young Scientists (B).




\vspace{4mm}
\noindent 
Tsuru University, 
3-8-1 Tahara Tsuru Yamanashi, 402-8555 Japan
%
%
%
\\
{\it E-mail address}: {\tt okano@tsuru.ac.jp} 

\end{document}